\documentclass[11pt]{article}
\usepackage{moreverb}
\usepackage{booktabs}
\usepackage{amsthm,amsmath,amssymb,bm,bbm,graphicx,natbib}
\usepackage{float}
\usepackage{subfigure}
\usepackage{enumerate}

\newcommand{\tE}{\mathbb{E}}
\newcommand{\tVar}{\text{var}}

\numberwithin{equation}{section}
\theoremstyle{plain}
\newtheorem{theorem}{Theorem}
\newtheorem{lemma}{Lemma}

\newtheorem{algorithm}{Algorithm}

\usepackage{authblk}

\title{Optimal Subsampling for Large Sample Ridge Regression}
\author[1]{Yunlu Chen}
\author[2]{Nan Zhang \footnote{ZHANG Nan, corresponding author, E-mail: zhangnan@fudan.edu.cn. The authors gratefully acknowledge the support by National Natural Science Foundation of China, Grant Number: 11690014; Science and Technology Commission of Shanghai Municipality, Grant Number: 17JC1420200.}}
\affil[1,2]{Fudan University, Shanghai 200433, China}

\date{}

\begin{document}


\maketitle

\begin{abstract}

Subsampling is a popular approach to alleviating the computational burden for analyzing massive datasets. Recent efforts have been devoted to various statistical models without explicit regularization. In this paper, we develop an efficient subsampling procedure for the large sample linear ridge regression. In contrast to the ordinary least square estimator, the introduction of the ridge penalty leads to a subtle trade-off between bias and variance. We first investigate the asymptotic properties of the subsampling estimator and then propose to minimize the asymptotic-mean-squared-error criterion for optimality. The resulting subsampling probability involves both ridge leverage score and $\ell_2$ norm of the predictor. To further reduce the computational cost for calculating the ridge leverage scores, we propose the algorithm with efficient approximation. We show by synthetic and real datasets that the algorithm is both statistically accurate and computationally efficient compared with existing subsampling based methods.
\end{abstract}








\section{Introduction}\label{int}
Linear regression is a popular method to depict the relationship between the response variable $y\in\mathcal Y$ and the covariate $\mathbf{x}\in\mathcal X\subset \mathbb R^p$. Observing $n$ independent and identically distributed data $\mathcal{F}_{n}=\{(\mathbf{x}_i,y_i)\}_{i=1}^n$, we consider the linear model $y_i=\mathbf{x}_i^{\top} \boldsymbol{\beta}+\epsilon_i,\, i=1,\dots,n,$
where $\epsilon_i$ is the independent and identically distributed error term with mean zero and variance $\sigma^2$. The model can be written in the matrix form
$$\mathbf{y}=\mathbf{X} \boldsymbol{\beta}+\mathbf{\epsilon},$$
where $\mathbf{y}=(y_1,\dots,y_n)^{\top}$ is a response vector, $\mathbf{X}=(\mathbf{x}_1,\dots,\mathbf{x}_n)^{\top}$ is an $n\times p$ design matrix, $\boldsymbol{\epsilon}=(\epsilon_1,\dots,\epsilon_n)^{\top}$ is an error vector. The ordinary least square approach minimizes $\|\mathbf{y}-\mathbf{X} \boldsymbol{\beta}\|^{2}$ and leads to $\widehat{\boldsymbol{\beta}}_{\text{OLS}}=\left(\mathbf{X}^{\top} \mathbf{X}\right)^{-1} \mathbf{X}^{\top} \mathbf{y}$ provided that $\mathbf{X}^{\top} \mathbf{X}$ is invertible. 
Ridge regression, proposed by \cite{hoerl1970ridge} 50 years ago, provides a remedy for ill-conditioned $\mathbf{X}^{\top} \mathbf{X}$ in computing the ordinary least square estimator. The ridge regression estimator is defined by adding a ridge on the diagonal of $\mathbf{X}^{\top} \mathbf{X}$, that is,
\begin{equation}\label{close_full}
    \widehat{\boldsymbol{\beta}}=\left(\mathbf{X}^{\top} \mathbf{X}+\lambda \mathbf{I}\right)^{-1} \mathbf{X}^{\top} \mathbf{y},
\end{equation}
where $\lambda>0$ is called the ridge parameter.

The optimization function for the ridge regression estimator can be written as 
\begin{equation}\label{opt_full}
    \min _{\boldsymbol{\beta}}\left\{\|\mathbf{y}-\mathbf{X} \boldsymbol{\beta}\|^{2}+\lambda \|\boldsymbol{\beta}\|^2\right\}.
\end{equation}
The ridge penalty introduces bias to the estimator, while the variance is reduced at the same time. It leads to a bias-variance trade-off when we attempt to predict at a new location, see e.g., \cite{hastie2009elements,hastie2020ridge}. Tuning the ridge parameter $\lambda$ is critical for balancing the bias and variance of the estimator. Typical methods for choosing the ridge parameter include the cross-validation and the generalized cross-validation \citep{golub1979generalized}.

With massive data, it is often computationally prohibitive to calculate the estimator when either the sample size or the dimension is super large. In recent years, many research efforts have been devoted to addressing the computational issue due to the large data matrix. \cite{kumar2012sampling} explored the sampling approach for the column subset selection problem by the Nyström method. \cite{derezinski2020improved} recently provided an improved theoretical guarantee for low-rank approximations of large datasets. 
Another popular idea in machine learning is coreset, which constructs estimators based on sub-data. \cite{kacham2020optimal} utilized the spectral graph sparsification result of \cite{batson2012twice} and proposed to merge the coresets obtained from multiple servers. \cite{mahoney2011randomized, woodruff2014sketching} studied matrix sketching to generate smaller datasets with random projections. \cite{wang2017sketched} addressed the statistical and algorithmic properties of classical sketch and Hessian sketch. Recently, under the context of ridge regression models, ridge leverage scores, introduced by \cite{alaoui2015fast}, are defined as the diagonal elements of matrix $\mathbf{X}\left(\mathbf{X}^{\top} \mathbf{X}+\lambda \mathbf{I}\right)^{-1} \mathbf{X}^{\top}$. \cite{cohen2017input} extended the concept and proposed a low-rank projection-based approach via ridge leverage score sampling.
\cite{homrighausen2020compressed} provided the approximated bias and variance for ridge regression but under the special condition that the compression matrix is of a sparse Bernoulli form.

Subsampling can be viewed as a special case of random projection or sketching. A general subsampling procedure is basically to first select a subsample from the original dataset according to certain subsampling probabilities and then construct an estimator via only the subsample. The efficiency of implementation and nice interpretability make subsampling-based methods attractive.
Based on the characteristics of the sampling step, existing methods can be summarized into two categories: deterministic and randomized subsampling. 
For the deterministic approach, \cite{wang2019information} proposed to select the subsample with extreme values on each dimension of $\mathbf X$ in linear regression such that the information matrix has a well-controlled determinant value. 
The second approach, the randomized subsampling, assigns subsampling probabilities to each observation and can achieve certain optimality by minimizing various criteria from the theory of experimental design.  \cite{drineas2011faster}, \cite{ma2015statistical} and \cite{zhu2015optimal} investigated the optimal subsampling for large sample linear regression via leverage scores, i.e., the diagonal elements of $\mathbf{X}\left(\mathbf{X}^{\top} \mathbf{X}\right)^{-1} \mathbf{X}^{\top}$. Such strategy has inspired further studies on versatile statistical models including logistic regression \citep{wang2018optimal}, quantile regression \citep{wang2020optimal} and generalized linear models \citep{ai2018optimal}. 

Our goal in this paper is to alleviate the computational burden for ridge regression with large-scale datasets. In particular, we focus on the case where the full sample size $n$ is much larger than the dimension $p$. Motivated from the idea of subsampling which concerns the asymptotic result \citep{zhu2015optimal, ma2015statistical}, we study the bias and variance of the regression coefficient estimator from the subsample. Taking the bias-variance trade-off into consideration, we propose to minimize the asymptotic-mean-squared-error criterion and show that the optimal subsampling probability for each observation depends on not only its ridge leverage score but also the $\ell_2$ norm of the covariate. 
Unlike existing subsampling methods for large sample regression models where no penalty term is involved, it plays an important role to select a proper ridge parameter. Although the derived optimal subsampling probabilities have explicit forms, it is unrealistic to directly apply them because quantities include the ridge parameter and the ridge leverage scores are computationally expensive to calculate. On the one hand, conventional methods for choosing the ridge parameter such as the cross-validation and the generalized cross-validation are time-consuming when applied to the full sample. However, based on the relationship between the best ridge parameter for the full sample and that for the subsample, we can instead apply the cross-validation on the subsample and extrapolate it to the full sample. 
On the other hand, for efficient approximation, we replace individual ridge leverage scores with their average. As a consequence, the optimal subsampling probabilities are proportional to the $\ell_2$ norms of predictors. Based on the aforementioned adjustments, our new method exhibits better performance with efficient computation than other sketching and subsampling algorithms over extensive simulation studies, especially when the subsample size is small.


The rest of the paper is organized as follows. Section \ref{methodology} presents the framework of the subsampling method and explains the details of ridge parameter selection and the optimal subsampling criterion. Section \ref{algorithm} proposes the optimal subsampling algorithm. Section \ref{simulation} and Section \ref{real} demonstrate the practical effectiveness of our algorithms via simulation and application, respectively.

\section{Methodology}\label{methodology}
\subsection{Subsampling framework}
To reduce the computation when dealing with datasets of large sample size $n$, the key step of a general subsampling procedure is to select a subsample of size $r\ll n$ from the original observations according to subsampling probabilities.
Extending the weighted estimation algorithm raised in \cite{ma2015statistical} to the ridge regression, we present the following framework for the ridge regression estimator $\widetilde{\boldsymbol{\beta}}$. Our proposed algorithms are based on this basic framework with its details shown in Section \ref{algorithm}.

\begin{enumerate}[Step 1.]
        \item Construct the subsampling probability for each sample $\{\pi_{i}\}_{i=1}^{n}$. Draw a subsample $(\mathbf{X}^*,\mathbf{y}^*)$ of size $r\ll n$ based on the probability. 
    \item Determine the ridge parameter $\widetilde{\lambda}$ for the subsample. Calculate the ridge regression estimator using the subsample, i.e.,
    \begin{equation}\label{weighted}
        \widetilde{\boldsymbol{\beta}}=\arg \min _{\boldsymbol{\beta}}\left\|\Phi^{*}\mathbf{y}^{*}-\Phi^{*}\mathbf{X}^{*} \boldsymbol{\beta}\right\|^{2}+\widetilde{\lambda}\|\boldsymbol{\beta}\|^2,
        \end{equation}
    where $\Phi^{*}=\operatorname{diag}\left\{1/\sqrt{r\pi_{k}^{*}}\}\right\}_{k=1}^{r}$.
\end{enumerate}

Under the above general subsampling framework, two key questions remain to be answered: 
\begin{enumerate}
    \item How to determine the ridge parameter for subsample $\widetilde{\lambda}$?
    \item What is the optimal subsampling probability for each sample $\{\pi_{i}\}_{i=1}^{n}$?
\end{enumerate}

\subsection{Ridge parameter selection}

The regularization in ridge regression plays an essential role in the prediction performance of the estimator. The ridge parameter $\lambda$ is usually unknown and requires careful tuning. Taking both the bias and variance into consideration, we can view the mean squared error as a function of the ridge parameter, and define the optimal ridge parameter as the one corresponding to the smallest mean squared error. For example, when the design matrix has orthonormal columns, the mean squared error can be derived as $p\sigma^2(1+\lambda)^{-2}+\lambda^2(1+\lambda)^{-2}\boldsymbol{\beta}^{\top}\boldsymbol{\beta}$, and thus the optimal ridge parameter is $p\sigma^2/{\boldsymbol{\beta}^{\top}\boldsymbol{\beta}}$. In practice, the cross-validation and its variants are applied to obtain the optimal ridge parameter. For $K$-fold cross-validation, the training data is divided into $K$ partitions $\{\mathbf{X}^{(k)},\mathbf{y}^{(k)}\}_{k=1}^{K}$ and we denote by $\widehat{\boldsymbol{\beta}}_{\backslash k}(\lambda)$ the estimated coefficient based on all partitions except the $k$th one. The optimal ridge parameter is $$\lambda_{\text{K-fold}}=\arg \min _{\lambda}K^{-1} \sum_{k=1}^{K}\|\mathbf{y}^{(k)}-\mathbf{X}^{(k)}\widehat{\boldsymbol{\beta}}_{\backslash k}(\lambda)\|^{2}.$$
The repeated fitting process by using different parts of the original sample leads to a high computational cost, especially when the sample size is large.
\cite{golub1979generalized} proposed the generalized cross-validation to reduce the computation cost of cross-validation. Consider the leave-one-out cross-validation, i.e., $K=n$. It can be shown that 
$$
\lambda_{\text{LOOCV}}
=\arg \min _{\lambda}n^{-1} \sum_{i=1}^{n}\left\{\frac{y_i-\mathbf{x}_i^{\top}  \widehat{\boldsymbol{\beta}}(\lambda)}{1-h_{ii}(\lambda)}\right\}^2,
$$
where $h_{i i}=\mathbf{x}_{i}^{\top}\left(\mathbf{X}^{\top} \mathbf{X}+\lambda \mathbf{I}\right)^{-1} \mathbf{x}_{i}$ is the diagonal element of the hat matrix $\mathbf{H}=\mathbf{X}\left(\mathbf{X}^{\top} \mathbf{X}+\lambda \mathbf{I}\right)^{-1} \mathbf{X}^{\top}$ where $i=1,\dots,n$. Quantity $h_{ii}$ measures the influential effect of the $i$th data points upon prediction and is called $\lambda$-leverage score \citep{alaoui2015fast}. The computational cost of calculating the ridge leverage score is $\mathcal{O}(np^2)$. 
The generalized cross-validation replaces individual leverage scores with their average $\textrm{tr}(\mathbf{H})$ to reduce computation, that is,
\begin{equation}\label{gcv}
  \lambda_{\text{GCV}}=\arg \min _{\lambda}n^{-1} \sum_{i=1}^{n}\left\{\frac{y_i-\mathbf{x}_i^{\top}  \widehat{\boldsymbol{\beta}}(\lambda)}{1-n^{-1} \textrm{tr}(\mathbf{H})}\right\}^{2}.  
\end{equation}

The optimal ridge parameter $\widetilde\lambda$ for the subsample can be estimated by minimizing the cross-validation criteria but over the subsample. We next provide the rationale of the approximation in the main theorem.

\subsection{Optimal subsampling}
Once the ridge parameter is fixed, we calculate subsampling probabilities for each observation, by which a subset of data points are selected from the full sample with replacement. We anticipate the estimator based on the subsample can achieve some optimality. In the ridge regression, the regularized estimator can perform much better than the ordinary least square estimator if the bias and variances are traded off properly. Therefore, we consider the mean-squared-error type of criterion which involves both bias and variance. 

We begin by investigating the difference between the subsampling estimator and its full-sample counterpart.
For convenience of analysis, we introduce some notations here to rewrite the subsampling estimator in the form concerning full data. Let $K_i$ be the number of times the observation $\mathbf{x}_{i}$ is sampled and $(K_1,\dots, K_n)$ thus follows a multinomial distribution. Let $\mathbf{W}=\mathbf{\Omega}\mathbf{K}$, where $\mathbf{K}=\operatorname{diag}\{K_i\}_{i=1}^n$, $\mathbf{\Omega}=\operatorname{diag}\{1/r\pi_i\}_{i=1}^n$. Simple algebra yields that the ridge regression estimator based on the subsample from \eqref{weighted} can be expressed as
\begin{equation}\label{close_sub}
    \widetilde{\boldsymbol{\beta}}=\left(\mathbf{X}^{*\top} \Phi^{*2} \mathbf{X}^*+\widetilde{\lambda} \mathbf{I}\right)^{-1} \mathbf{X}^{*\top} \Phi^{*2} \mathbf{y}^*
    =\left(\mathbf{X}^{\top} \mathbf{W} \mathbf{X}+\widetilde{\lambda} \mathbf{I}\right)^{-1} \mathbf{X}^{\top}\mathbf{W} \mathbf{y}.
\end{equation} 
In the following lemma, we demonstrate the difference between the estimator \eqref{close_sub} and the full-sample estimator \eqref{close_full}.

\begin{lemma}\label{lemma1}
If $0<\pi_i<1,i=1,\dots,n$,
$(\mathbf{X}^{\top}\mathbf{X}+\lambda\mathbf{I})^{-1}\mathbf{X}^{\top}(\mathbf{W}-\mathbf{I})\mathbf{X}=\mathcal{O}_{p}(r^{-1/2})$ and $(\widetilde{\lambda}-\lambda)(\mathbf{X}^{\top}\mathbf{X}+\lambda\mathbf{I})^{-1}=\mathcal{O}(r^{-1/2})$, then
\begin{equation}\label{biasvar}
    \widetilde{\boldsymbol{\beta}}-\widehat{\boldsymbol{\beta}}= (\mathbf{X}^{\top}\mathbf{X}+\lambda\mathbf{I})^{-1}\mathbf{X}^{\top}\mathbf{W}\mathbf{e}-\widetilde{\lambda}(\mathbf{X}^{\top}\mathbf{X}+\lambda\mathbf{I})^{-1}\widehat{\boldsymbol{\beta}}+\mathcal{O}_{p}(r^{-1}),
\end{equation}
where $\mathbf{e}=\mathbf{y}-\mathbf{X}\widehat{\boldsymbol{\beta}}$. 
\end{lemma}

\begin{proof}
We first rewrite the subsampling estimator by multiplying $(\mathbf{X}^{\top}\mathbf{X}+\lambda\mathbf{I})(\mathbf{X}^{\top}\mathbf{X}+\lambda\mathbf{I})^{-1}$,
\begin{equation}\label{tay}
    \widetilde{\boldsymbol{\beta}}=\left\{(\mathbf{X}^{\top}\mathbf{X}+\lambda\mathbf{I})^{-1}(\mathbf{X}^{\top}\mathbf{W}\mathbf{X}+\widetilde{\lambda}\mathbf{I})\right\}^{-1}\left\{(\mathbf{X}^{\top}\mathbf{X}+\lambda\mathbf{I})^{-1}\mathbf{X}^{\top}\mathbf{W}\mathbf{y}\right\}.
\end{equation}
For the inverse term, we apply the Taylor series expansion,
\begin{align*}
\left\{(\mathbf{X}^{\top}\mathbf{X}+\lambda\mathbf{I})^{-1}(\mathbf{X}^{\top}\mathbf{W}\mathbf{X}+\widetilde{\lambda}\mathbf{I})\right\}^{-1}&=\left[\mathbf{I}+(\mathbf{X}^{\top}\mathbf{X}+\lambda\mathbf{I})^{-1}\{\mathbf{X}^{\top}(\mathbf{W}-\mathbf{I})\mathbf{X}+(\widetilde{\lambda}-\lambda)\mathbf{I}\}\right]^{-1}\\
&=\mathbf{I}-(\mathbf{X}^{\top}\mathbf{X}+\lambda\mathbf{I})^{-1}\mathbf{X}^{\top}(\mathbf{W}-\mathbf{I})\mathbf{X}-(\widetilde{\lambda}-\lambda)(\mathbf{X}^{\top}\mathbf{X}+\lambda\mathbf{I})^{-1}+\mathcal{O}_{p}(r^{-1}).
\end{align*}
For the other term in \eqref{tay}, 
\begin{align*}
    (\mathbf{X}^{\top}\mathbf{X}+\lambda\mathbf{I})^{-1}\mathbf{X}^{\top}\mathbf{W}\mathbf{y} &=(\mathbf{X}^{\top}\mathbf{X}+\lambda\mathbf{I})^{-1}\left\{\mathbf{X}^{\top}\mathbf{y}+\mathbf{X}^{\top}(\mathbf{W}-\mathbf{I})\mathbf{y}\right\}\\
    &=\widehat{\boldsymbol{\beta}}+(\mathbf{X}^{\top}\mathbf{X}+\lambda\mathbf{I})^{-1}\mathbf{X}^{\top}(\mathbf{W}-\mathbf{I})\mathbf{y}.
\end{align*}
Since $(\mathbf{X}^{\top}\mathbf{X}+\lambda\mathbf{I})^{-1}\mathbf{X}^{\top}(\mathbf{W}-\mathbf{I})\mathbf{y}$ and $(\mathbf{X}^{\top}\mathbf{X}+\lambda\mathbf{I})^{-1}\mathbf{X}^{\top}(\mathbf{W}-\mathbf{I})\mathbf{e}$ are of the same order as $(\mathbf{X}^{\top}\mathbf{X}+\lambda\mathbf{I})^{-1}\mathbf{X}^{\top}(\mathbf{W}-\mathbf{I})\mathbf{X}$,
\begin{align*}
    \widetilde{\boldsymbol{\beta}}&=\widehat{\boldsymbol{\beta}}+(\mathbf{X}^{\top}\mathbf{X}+\lambda\mathbf{I})^{-1}\mathbf{X}^{\top}(\mathbf{W}-\mathbf{I})\mathbf{e}+(\lambda-\widetilde{\lambda})(\mathbf{X}^{\top}\mathbf{X}+\lambda\mathbf{I})^{-1}\widehat{\boldsymbol{\beta}}+\mathcal{O}_{p}(r^{-1})\\
    &=\widehat{\boldsymbol{\beta}}+(\mathbf{X}^{\top}\mathbf{X}+\lambda\mathbf{I})^{-1}\mathbf{X}^{\top}\mathbf{W}\mathbf{e}-\widetilde{\lambda}(\mathbf{X}^{\top}\mathbf{X}+\lambda\mathbf{I})^{-1}\widehat{\boldsymbol{\beta}}+\mathcal{O}_{p}(r^{-1})
\end{align*}
The second equation holds due to the normal equation for ridge regression.
\end{proof}


We investigate the asymptotic mean squared error of $\widetilde{\boldsymbol{\beta}}$, which is used as our criterion for determining the optimal subsampling probabilities. 
\begin{theorem}\label{them1}
If the full sample size $n$ is fixed, $\|\mathbf{x_i}\|<\infty$, $i=1,\dots,n$, the sampling probabilities $\{\pi_i\}_{i=1}^n$ are nonzero, and $\widetilde{\lambda}-\lambda=\mathcal{O}(r^{-1/2})$,
then the asymptotic variance and mean are
\begin{enumerate}
    \item $\text{AVar}(\widetilde{\boldsymbol{\beta}})= \Sigma_c-\lambda^2 r^{-1}(\mathbf{X}^{\top}\mathbf{X}+\lambda\mathbf{I})^{-1} \widehat{\boldsymbol{\beta}} \widehat{\boldsymbol{\beta}}^{\top} (\mathbf{X}^{\top}\mathbf{X}+\lambda\mathbf{I})^{-1}$, where \\
    $\Sigma_c=r^{-1}(\mathbf{X}^{\top}\mathbf{X}+\lambda\mathbf{I})^{-1}(\sum_{i=1}^n \pi_i^{-1} e_i^2\mathbf{x}_{i}  \mathbf{x}_{i}^{\top})(\mathbf{X}^{\top}\mathbf{X}+\lambda\mathbf{I})^{-1}$, $e_i=y_i-\mathbf{x}_i^{\top}\widehat{\boldsymbol{\beta}}$.\label{them11}
    \item $\text{AE}(\widetilde{\boldsymbol{\beta}})=\widehat{\boldsymbol{\beta}}+(\lambda-\widetilde{\lambda})(\mathbf{X}^{\top}\mathbf{X}+\lambda\mathbf{I})^{-1} \widehat{\boldsymbol{\beta}}$.\label{them12}
\end{enumerate}
The asymptotic mean squared error of $\widetilde{\boldsymbol{\beta}}$ is therefore $$\text{AMSE}(\widetilde{\boldsymbol{\beta}})=\Sigma_c-\lambda^2 r^{-1}(\mathbf{X}^{\top}\mathbf{X}+\lambda\mathbf{I})^{-1} \widehat{\boldsymbol{\beta}} \widehat{\boldsymbol{\beta}}^{\top} (\mathbf{X}^{\top}\mathbf{X}+\lambda\mathbf{I})^{-1}+(\lambda-\widetilde{\lambda})^2(\mathbf{X}^{\top}\mathbf{X}+\lambda\mathbf{I})^{-1} \widehat{\boldsymbol{\beta}} \widehat{\boldsymbol{\beta}}^{\top} (\mathbf{X}^{\top}\mathbf{X}+\lambda\mathbf{I})^{-1}.$$

\end{theorem}
\begin{proof}
We begin by the deduction of the result~\ref{them11}. The roadmap of the proof of this part is motivated from the result in \cite{zhu2015optimal} since the variance term in linear regression case is of similar form. We use Cramer-Wold device to establish the asymptotic normality of $(\mathbf{X}^{\top}\mathbf{X}+\lambda\mathbf{I})^{-1}\mathbf{X}^{\top}\mathbf{W}\mathbf{e}=(\mathbf{X}^{\top}\mathbf{X}+\lambda\mathbf{I})^{-1}\sum_{j=1}^r\mathbf{X}^{\top}\mathbf{\Omega}\mathbf{K}^{(j)}\mathbf{e}$. Consider each term in the summation, for any non-zero constant vector $\mathbf{b}\in \mathbb R^p$, we have
\begin{equation*}
    \tVar\{(\mathbf{b}^{\top}\mathbf{X}^{\top}\mathbf{X}+\lambda\mathbf{I})^{-1}\mathbf{X}^{\top}\mathbf{\Omega}\mathbf{K}^{(1)}\mathbf{e}\}=r^{-1}\mathbf{a}^{\top}\left(\sum_{i=1}^n \frac{e_i^2}{\pi_i}\mathbf{x}_i\mathbf{x}_i^T\right)\mathbf{a}-r^{-1}\mathbf{a}^{\top}\mathbf{X}^{\top}\mathbf{e}\mathbf{e}^{\top}\mathbf{X}\mathbf{a},
\end{equation*}
where $\mathbf{a}=(\mathbf{X}^{\top}\mathbf{X}+\lambda\mathbf{I})^{-1}\mathbf{b}$. By applying the normal equation of ridge regression, we have
\begin{equation*}
    \tVar\{(\mathbf{b}^{\top}\mathbf{X}^{\top}\mathbf{X}+\lambda\mathbf{I})^{-1}\mathbf{X}^{\top}\mathbf{\Omega}\mathbf{K}^{(1)}\mathbf{e}\}=\mathbf{b}^{\top}\left\{\Sigma_c+\lambda^2 r^{-1}(\mathbf{X}^{\top}\mathbf{X}+\lambda\mathbf{I})^{-1} \widehat{\boldsymbol{\beta}} \widehat{\boldsymbol{\beta}}^{\top} (\mathbf{X}^{\top}\mathbf{X}+\lambda\mathbf{I})^{-1}\right\}\mathbf{b}.
\end{equation*}
By using Lindeberg–Lévy CLT, we have the variance of the summation,
\begin{equation*}
    \tVar\{(\mathbf{b}^{\top}(\mathbf{X}^{\top}\mathbf{X}+\lambda\mathbf{I})^{-1}\sum_{j=1}^r\mathbf{X}^{\top}\mathbf{\Omega}\mathbf{K}^{(j)}\mathbf{e}\}=\mathbf{b}^{\top}\left\{\Sigma_c+\lambda^2 r^{-1}(\mathbf{X}^{\top}\mathbf{X}+\lambda\mathbf{I})^{-1} \widehat{\boldsymbol{\beta}} \widehat{\boldsymbol{\beta}}^{\top} (\mathbf{X}^{\top}\mathbf{X}+\lambda\mathbf{I})^{-1}\right\}\mathbf{b}
\end{equation*}
Therefore, we have the result~\ref{them11} due to the Cramer-Wold device.

Then consider the result~\ref{them12}. Since $\mathbf{W}=\mathbf{\Omega}\mathbf{K}$, where each element $K_i$ in $\mathbf{K}=\operatorname{diag}\{K_i\}_{i=1}^n$ follows the multinomal distribution $\text{Mult}(r,\{\pi_i\}_{i=1}^n)$, then $\tE(W_i)=1$, where $W_i$ is the diagonal element of matrix $\mathbf{W}$. Thus, we can calculate the expectation
\begin{align*}
    \tE\{(\mathbf{X}^{\top}\mathbf{X}+\lambda\mathbf{I})^{-1}\mathbf{X}^{\top}\mathbf{W}\mathbf{e}\}&=(\mathbf{X}^{\top}\mathbf{X}+\lambda\mathbf{I})^{-1}\mathbf{X}^{\top}\mathbf{e}\\
    &=\lambda(\mathbf{X}^{\top}\mathbf{X}+\lambda\mathbf{I})^{-1}\widehat{\boldsymbol{\beta}},
\end{align*}
with the second equation following the normal equation. Consequently, we have $\text{AE}(\widetilde{\boldsymbol{\beta}})=\widehat{\boldsymbol{\beta}}+(\lambda-\widetilde{\lambda})(\mathbf{X}^{\top}\mathbf{X}+\lambda\mathbf{I})^{-1} \widehat{\boldsymbol{\beta}}$.
\end{proof}
%
The above theorem shows that the controllable part of the criterion is included in the variance term. Considering the expression of variance, only $\Sigma_c$ depends on the subsampling probability $\{\pi_i\}_{i=1}^n$. We resort to minimize the expected trace of $\Sigma_c$ to obtain the corresponding subsampling probability, as shown in the following theorem.
\begin{theorem}\label{them2}
When $$\pi_i=\frac{\sqrt{\left(1-h_{i i}\right)}\left\|\mathbf{x}_{i}\right\|}{\sum_{j=1}^{n} \sqrt{\left(1-h_{j j}\right)}\left\|\mathbf{x}_{j}\right\|},$$
$\tE\{\textrm{tr}(\Sigma_c)\}$ attains its minimum, where $h_{ii}$ is the ridge leverage score, $i=1,\dots,n$.
\end{theorem}

\begin{proof}
Since $\tE\{\textrm{tr}(\Sigma_c)\}=r^{-1} \sum_{i=1}^{n}\pi_{i}^{-1} (1-h_{i i})\left\|\mathbf{x}_{i}\right\|^{2}$, we can get the following result by applying Hölder's inequality,
\begin{align*}
    r^{-1} \sum_{i=1}^{n} \frac{1-h_{i i}}{\pi_{i}}\left\|\mathbf{x}_{i}\right\|^{2}
    =r^{-1} \sum_{i=1}^{n} \frac{1-h_{i i}}{\pi_{i}}\left\|\mathbf{x}_{i}\right\|^{2} \sum_{i=1}^{n} \pi_{i}
    \geq \left\{\sum_{i=1}^{n} \sqrt{\left(1-h_{i i}\right)\left\|\mathbf{x}_{i}\right\|^{2}}\right\}^{2},
\end{align*}
with the equality holds if and only if $\pi_i\propto\sqrt{\left(1-h_{i i}\right)}\left\|\mathbf{x}_{i}\right\|$.
\end{proof}
In Theorem~\ref{them2}, we obtain the optimal subsampling probability for each observation. It involves both the ridge leverage score and the $\ell_2$ norm of the predictor. Our subsampling strategy is different from the sketching scheme for ridge regression \citep{cohen2017input} which only utilized the ridge leverage score. We will compare these methods on simulation and real data.

\section{Algorithm}\label{algorithm}
Based on the deduction in the above section, we obtain the subsampling probability and the approaching rate between the ridge parameter for subsample $\widetilde{\lambda}$ and that for full sample $\lambda$. Integrating these ingredients with the general subsampling procedure, $\widetilde{\lambda}=\lambda$ and $\pi_{i}=\frac{\sqrt{\left(1-h_{i i}\right)}\left\|\mathbf{x}_{i}\right\|}{\sum_{j=1}^{n} \sqrt{\left(1-h_{j j}\right)}\left\|\mathbf{x}_{j}\right\|},i=1,\dots,n$, we can calculate the subsampling estimator. However, there are two quantities whose calculations are still computationally demanding. First, it requires $\mathcal{O}(np^2)$ to compute the exact ridge leverage scores, which amounts to the same computation as the full-sample estimator. Second, we need to calculate $\lambda$ first for calculating the ridge leverage score. The ridge parameter $\widetilde{\lambda}$ is then set as $\lambda$, which is not desirable since applying the cross-validation to choose $\lambda$ is time-consuming. 

To address the aforementioned issues, we propose an efficient approximation to the optimal subsampling probabilities in Theorem~\ref{them2}. 
Similar to the idea of generalized cross-validation in \eqref{gcv}, we approximate the individual ridge leverage score with their average, i.e., $n^{-1} \textrm{tr}(\mathbf{H})$. It corresponds to the scenario where ridge leverage scores are not highly heterogeneous. Therefore, the subsampling probability reduces to $\pi_i=\|\mathbf{x}_i\|/\sum_{j=1}^n \|\mathbf{x}_j\|,i=1,\dots,n$, which involves only the $\ell_2$ norm of the predictors. 
Moreover, such an approximation of the subsampling probability no longer depends on the ridge leverage score, and hence we do not need the full sample $\lambda$ to calculate $\pi_i$. In this way, we can perform the cross-validation or generalized cross-validation to directly calculate $\widetilde{\lambda}$ with the selected subsample at a much lower computational cost. The optimal subsampling ridge regression estimation we propose is summarized in Algorithm~\ref{alg}.

\begin{algorithm}\label{alg}
Ridge Regression with Optimal Subsampling
\begin{enumerate}[Step 1.] 
    \item Construct the subsampling probability for each sample $\pi_{i}=\|\mathbf{x}_i\|/\sum_{j=1}^n \|\mathbf{x}_j\|,i=1,\dots,n$. Draw a subsample $(\mathbf{X}^*,\mathbf{y}^*)$ of size $r\ll n$ based on the probability. \label{algstep1}
    \item Calculate the ridge regression estimator
    \begin{equation}\label{weighted2}
        \widetilde{\boldsymbol{\beta}}=\arg \min _{\boldsymbol{\beta}}\left\|\Phi^{*}\mathbf{y}^{*}-\Phi^{*}\mathbf{X}^{*} \boldsymbol{\beta}\right\|^{2}+\widetilde{\lambda}\|\boldsymbol{\beta}\|^2,
        \end{equation}
        where $\widetilde{\lambda}$ is selected via cross-validation for response vector $\Phi^{*}\mathbf{y}^{*}$ and design matrix $\Phi^{*}\mathbf{X}^{*}$ and $\Phi^{*}=\operatorname{diag}\left\{1/\sqrt{r\pi_{k}^{*}}\}\right\}_{k=1}^{r}$. \label{algstep2}
\end{enumerate}
\end{algorithm}
Compared to calculating the exact ridge leverage score, the calculation of subsampling probabilities in Step~\ref{algstep1} of Algorithm~\ref{alg} avoids accessing the full design matrix $\mathbf X$, and hence the $\ell_2$ norm of predictors can be obtained in parallel.
According to the closeness between $\widetilde\lambda$ and $\lambda$ revealed in Theorem~\ref{them1}, we can extrapolate $\widetilde\lambda$ for the subsample to estimate $\lambda$ for the full sample. It allows us to confirm our theoretical findings with numerical experiments which are presented in the following sections.

\section{Simulation}\label{simulation}
In the simulation study, we begin by demonstrating the effectiveness of the approximation of the ridge leverage score.
We then compare the proposed methods with other subsampling approaches developed for large sample ridge regression or linear regression, including the ridge leverage score subsampling \citep{cohen2017input}, uniform subsampling for ridge regression, optimal subsampling for linear regression \citep{zhu2015optimal} and the information-based optimal subdata selection for linear regression (IBOSS) \citep{wang2019information}. Simulated data are generated in six settings. In each simulation, we generate the full data set of size $n=10^5$ and dimension $p=50$. Subsample sizes are set as $r=100,200,400,800,1600,3200,6400$. The design matrix $\mathbf{X}$ is standardized before being fed into the model. Each experiment is repeated $20$ times. We use the mean squared error (MSE) of the estimated coefficient $\widetilde{\boldsymbol{\beta}}$ to evaluate the performance.

The errors $\epsilon_i, i=1,\dots,n$ are independently and identically generated from $N(0,9)$. We set the simulations as follows. For a $q<p$, let $\mathbf{\Sigma}$ be the $q\times q$ covariance matrix with element $\mathbf{\Sigma}_{i,j}=0.5^{1(i \neq j)},i,j=1,\dots,n$. Consider $q$-dimensional $\mathbf{x}_i\sim N(\mathbf{0}, \mathbf{\Sigma}),i=1,\dots,n$ for the true linear model with $\boldsymbol{\beta}=\mathbf{1}_{q\times 1}$. An additional $(p-q)$-dimensional term $\mathbf{x}^a$ is generated without being used in the true model, since we want to test if the subsample helps identify the appropriate relationship between the responses and the true covariates. 
\begin{enumerate}[\text{Case} 1.]
    \item $q=10$, $\mathbf{X}^a_i$ follows a multivariate normal distribution, where columns of $\mathbf{X}^a_{i}$ are i.i.d. samples from $N(0,1)$.
    \item $q=10$, $\mathbf{X}^a_i$ follows a multivariate lognormal distribution, where columns of $\mathbf{X}^a_{i}$ are i.i.d. samples from $LN(0,1)$.
    \item $q=10$, $\mathbf{X}^a_i$ follows a multivariate $t$ distribution with degrees of freedom $2$, where columns of $\mathbf{X}^a_{i}$ are i.i.d. samples from $t_2(0,1)$.
    \item $q=25$, $\mathbf{X}^a_i$ follows a multivariate normal distribution, where columns of $\mathbf{X}^a_{i}$ are i.i.d. samples from $N(0,1)$.
    \item $q=25$, $\mathbf{X}^a_i$ follows a multivariate lognormal distribution, where columns of $\mathbf{X}^a_{i}$ are i.i.d. samples from $LN(0,1)$.
    \item $q=25$, $\mathbf{X}^a_i$ follows a multivariate $t$ distribution with degrees of freedom $2$, where columns of $\mathbf{X}^a_{i}$ are i.i.d. samples from $t_2(0,1)$.
\end{enumerate}

First, we show that subsampling by using the fast approximation of ridge leverage score (ROPT) is similar to that by applying the accurate one (ROPT-acc). Figure~\ref{simu_pi} displays the comparison result of MSE of the estimators by using the two sampling probabilities. Both methods have similar performance in all 6 cases given different subsample sizes. Therefore, the effectiveness of the $\ell_2$ approximation of ridge leverage score is demonstrated. \begin{figure}[h!]
\centering
\subfigure{
\begin{minipage}[b]{0.31\linewidth}
\includegraphics[width=1\linewidth]{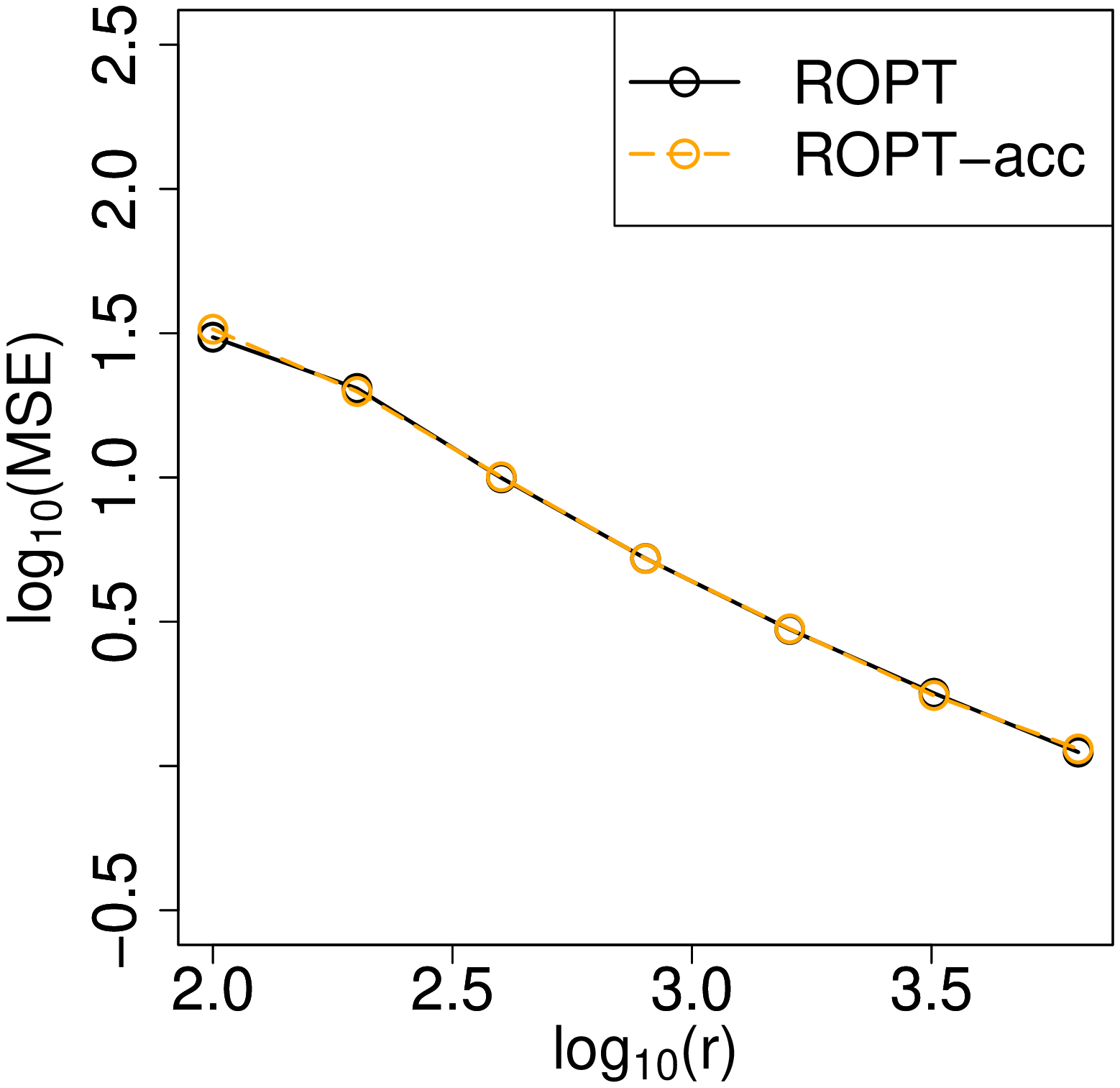}\vspace{4pt}
\includegraphics[width=1\linewidth]{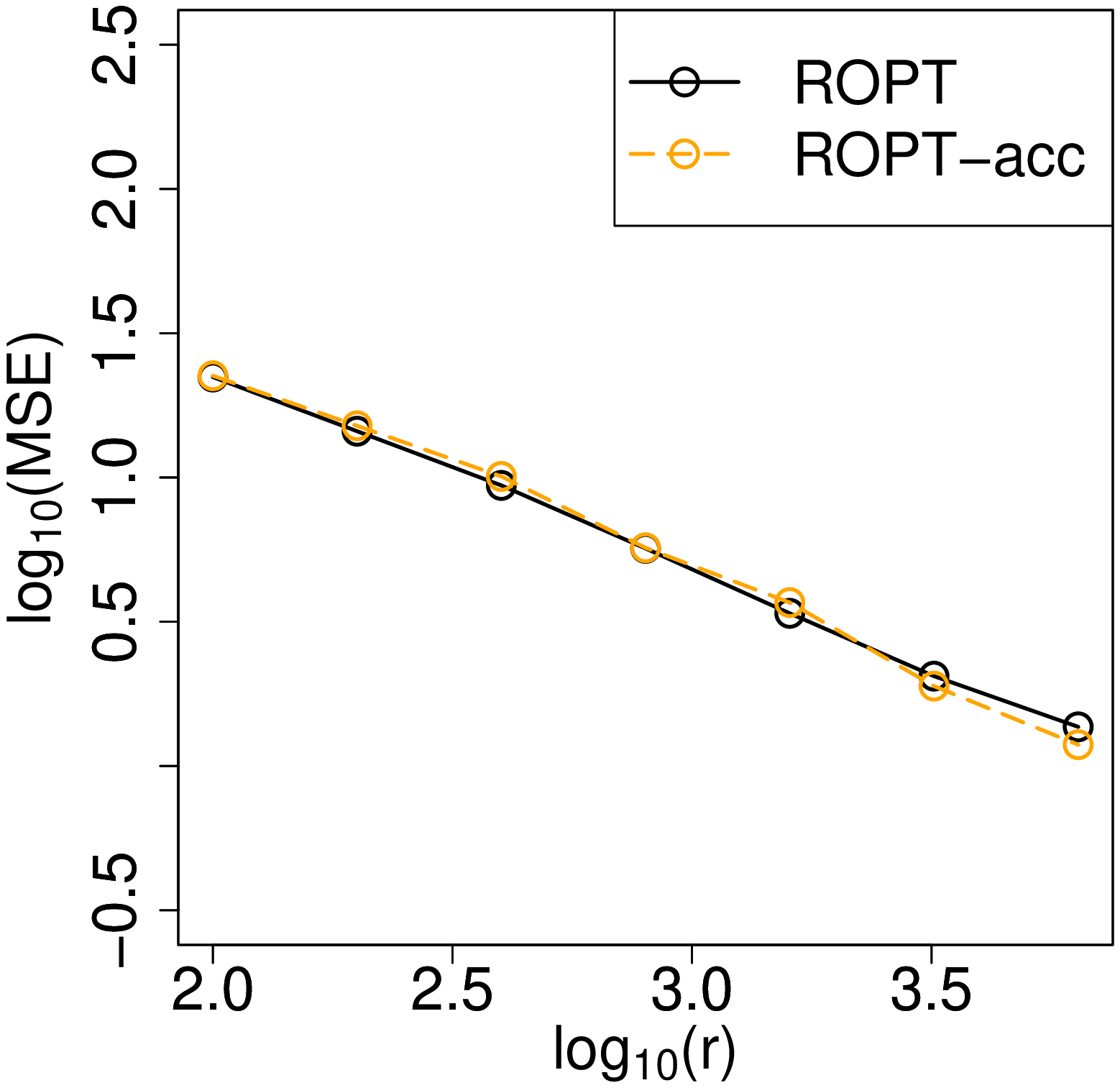}
\end{minipage}}
\subfigure{
\begin{minipage}[b]{0.31\linewidth}
\includegraphics[width=1\linewidth]{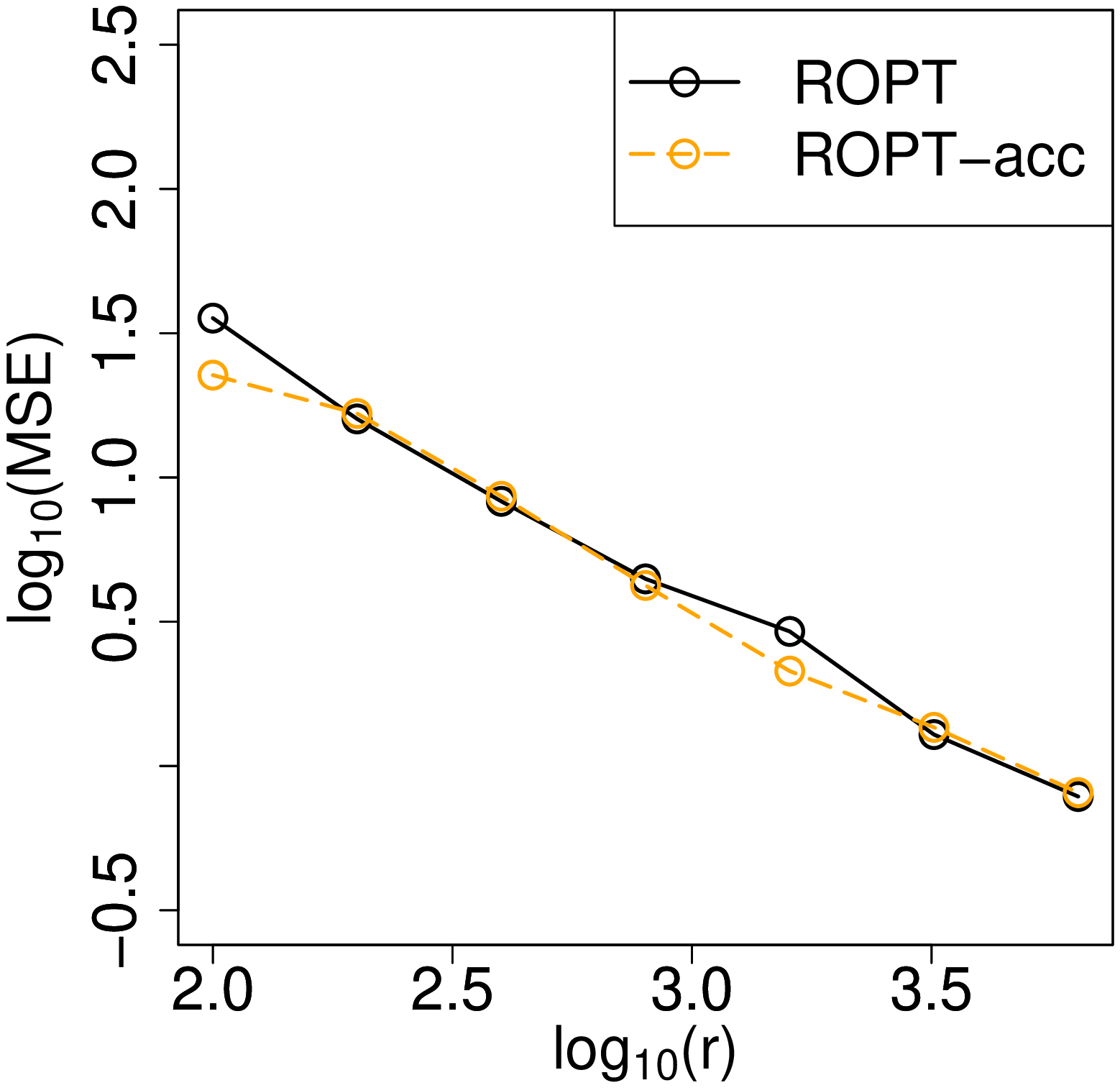}\vspace{4pt}
\includegraphics[width=1\linewidth]{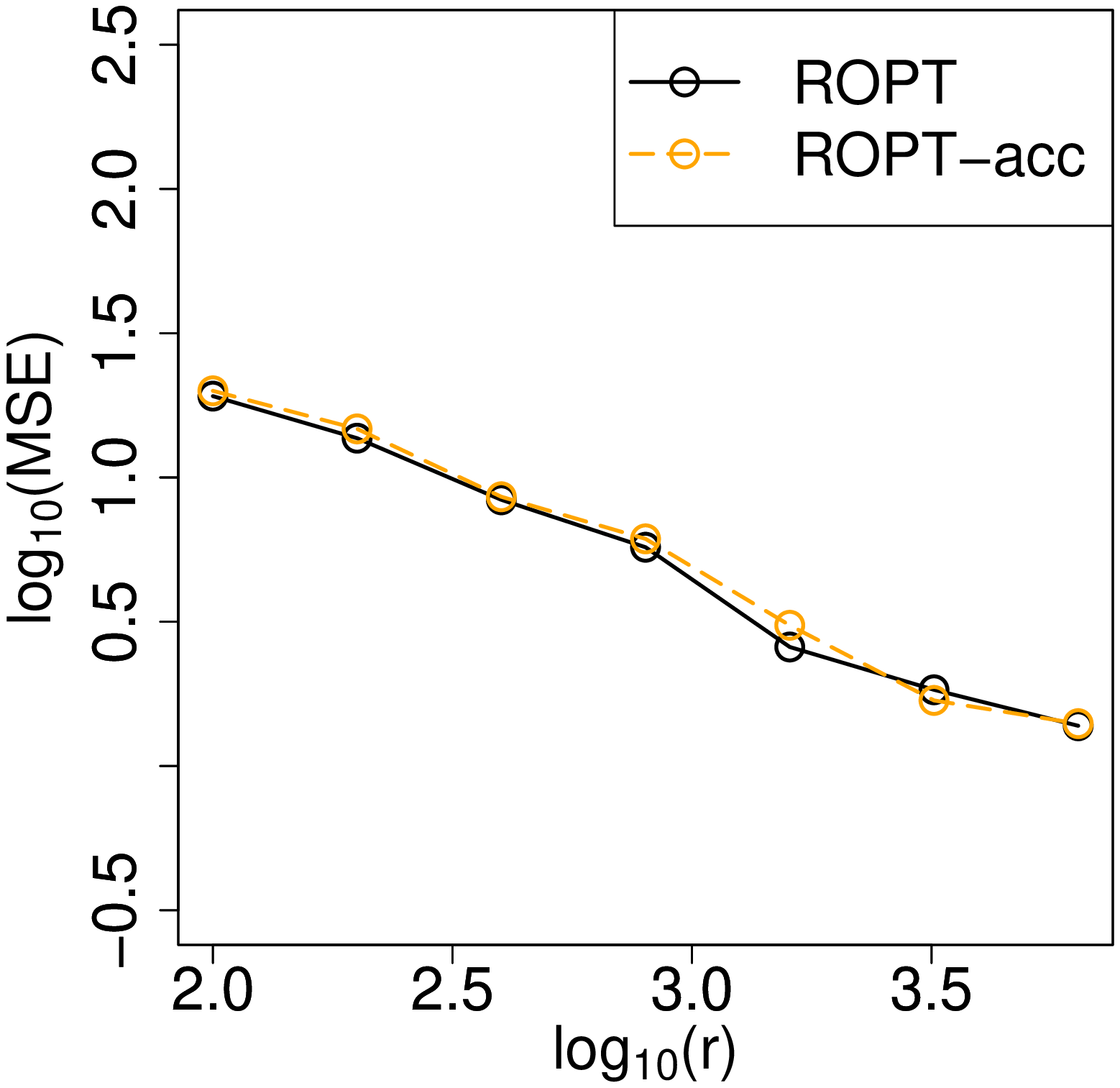}
\end{minipage}}
\subfigure{
\begin{minipage}[b]{0.31\linewidth}
\includegraphics[width=1\linewidth]{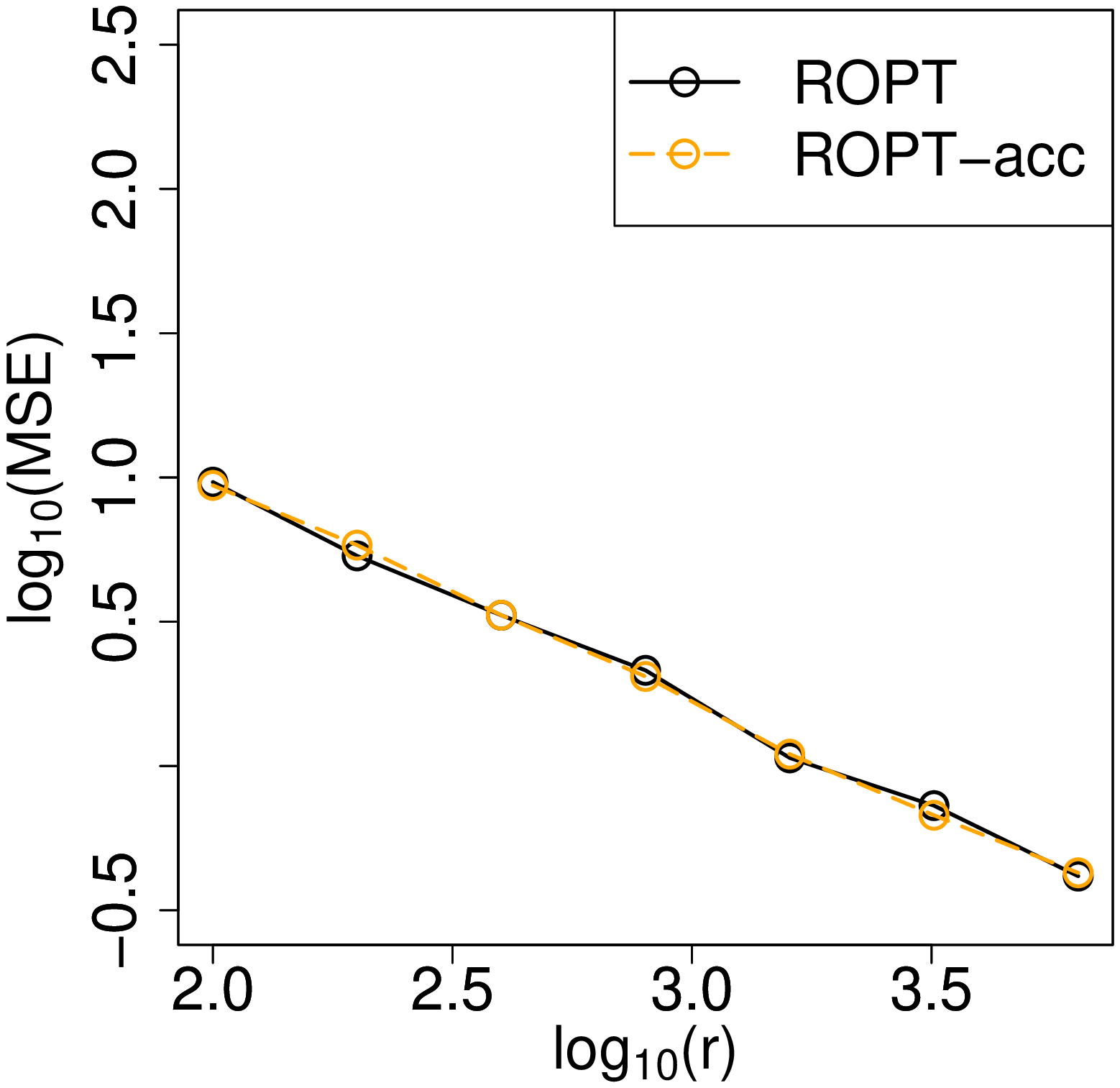}\vspace{4pt}
\includegraphics[width=1\linewidth]{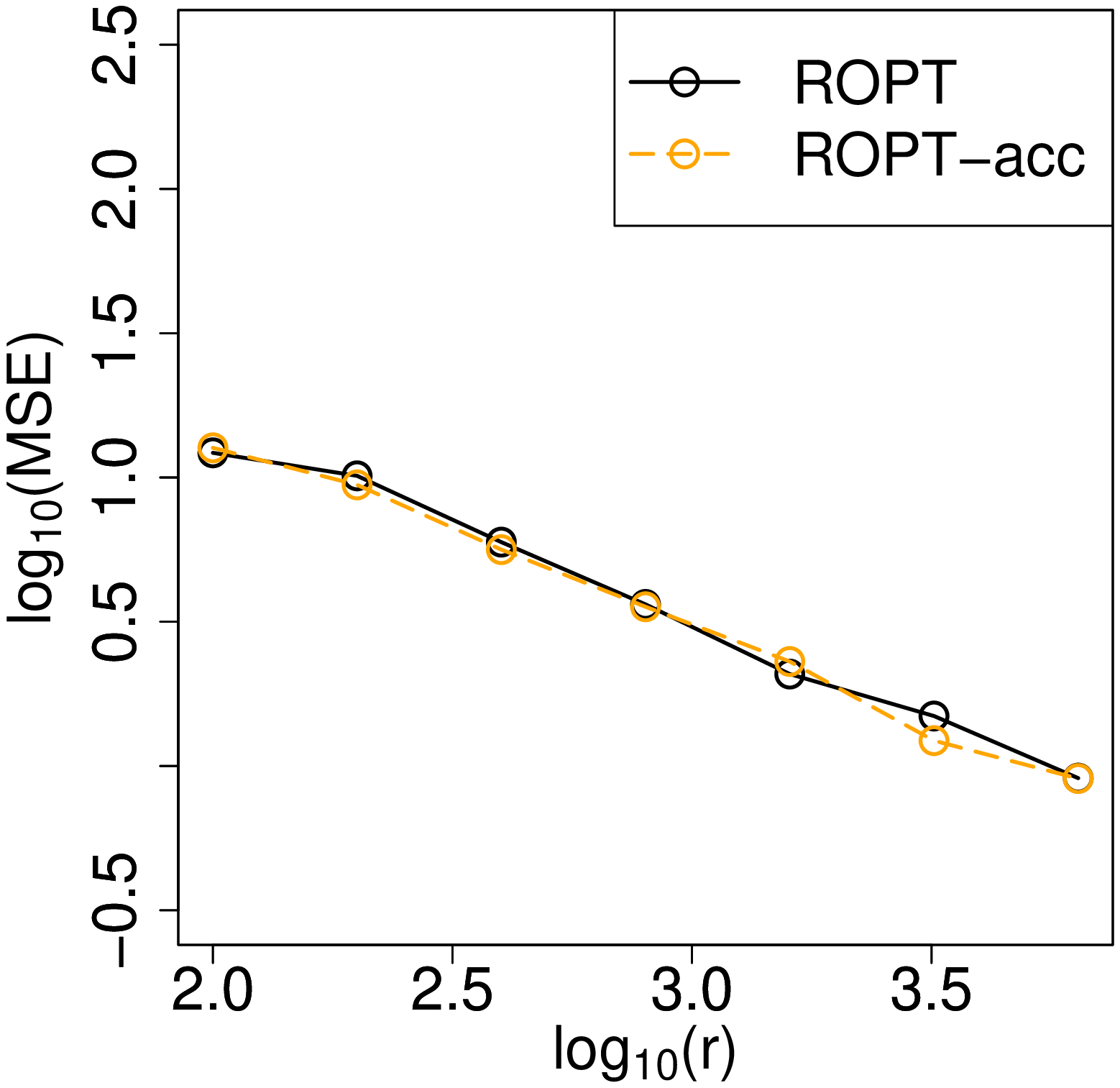}
\end{minipage}}
\caption{Comparison of different subsampling probabilities: $x$-axis is the logarithm of the subsample size, $y$-axis is the logarithm of the mean squared error of $\widetilde{\boldsymbol{\beta}}$.}
\label{simu_pi}
\end{figure}

Then we compare our proposed method with the ridge leverage score subsampling (RLEV), uniform subsampling for ridge regression (RUNIF), optimal subsampling for linear regression (OPT), and the information-based optimal subdata selection for linear regression (IBOSS). 
The first two competitors are raised for the ridge regression while the rest two are proposed for the linear regression. 
\begin{figure}[h!]
\centering
\subfigure{
\begin{minipage}[b]{0.31\linewidth}
\includegraphics[width=1\linewidth]{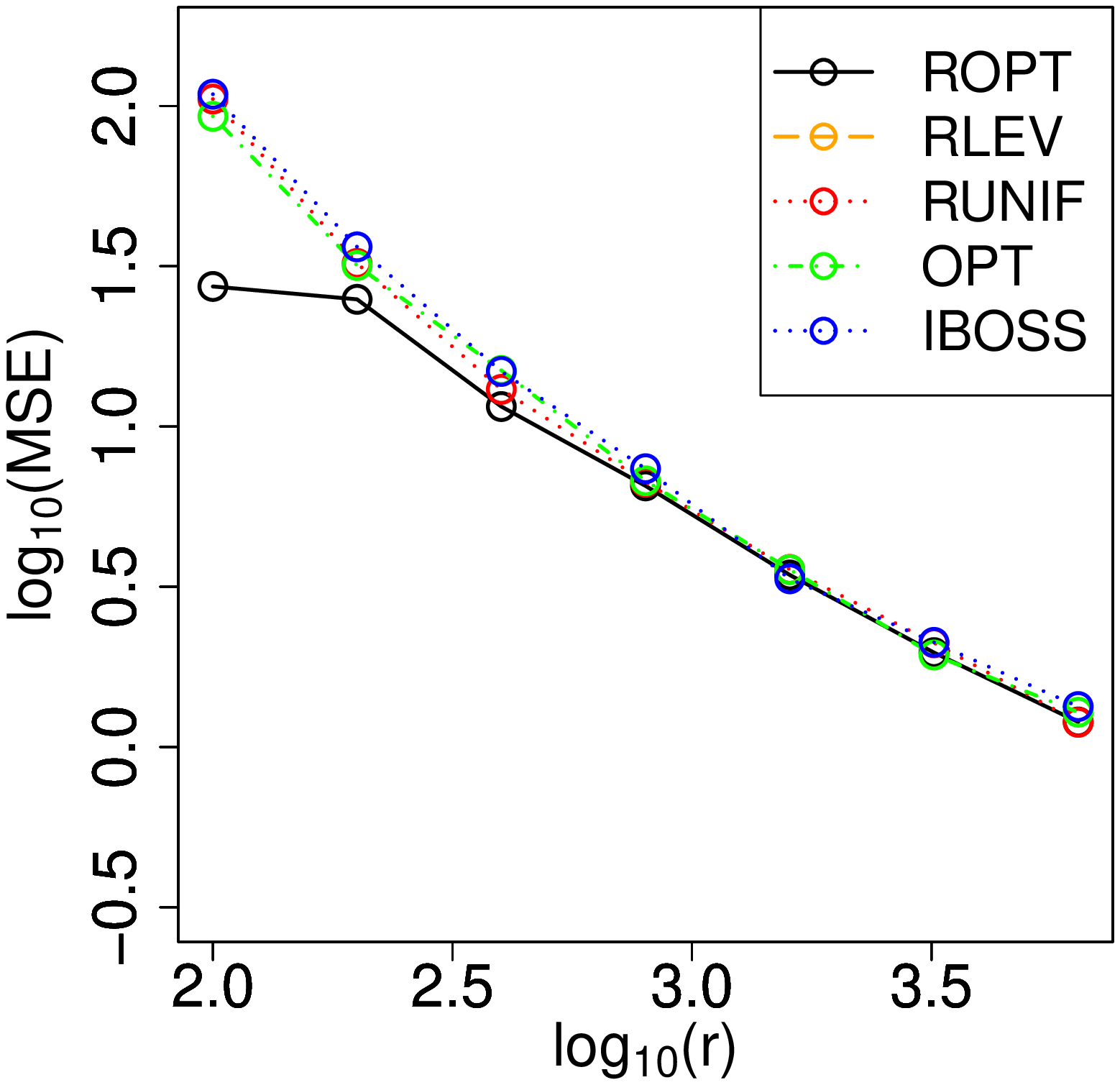}\vspace{4pt}
\includegraphics[width=1\linewidth]{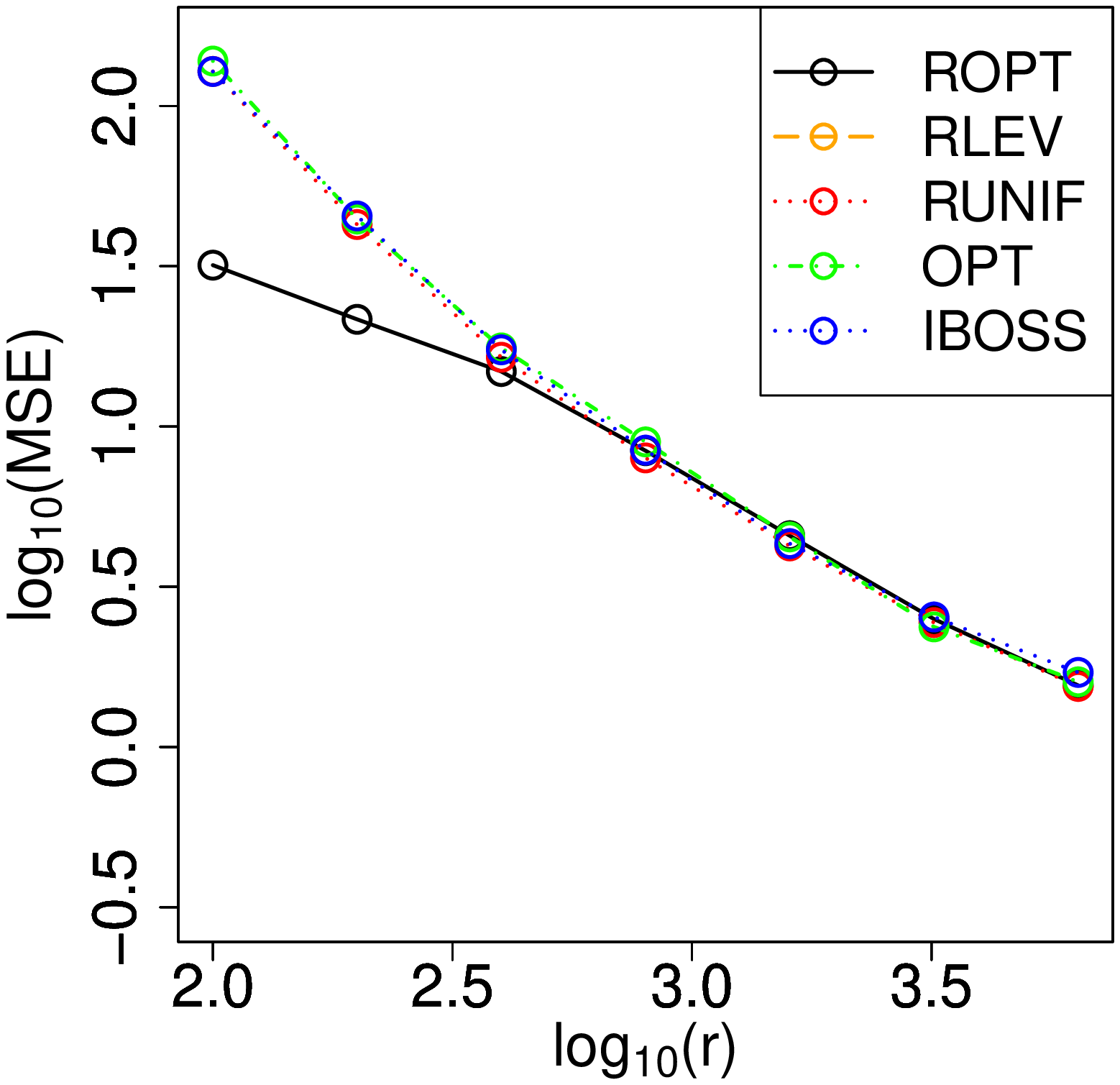}
\end{minipage}}
\subfigure{
\begin{minipage}[b]{0.31\linewidth}
\includegraphics[width=1\linewidth]{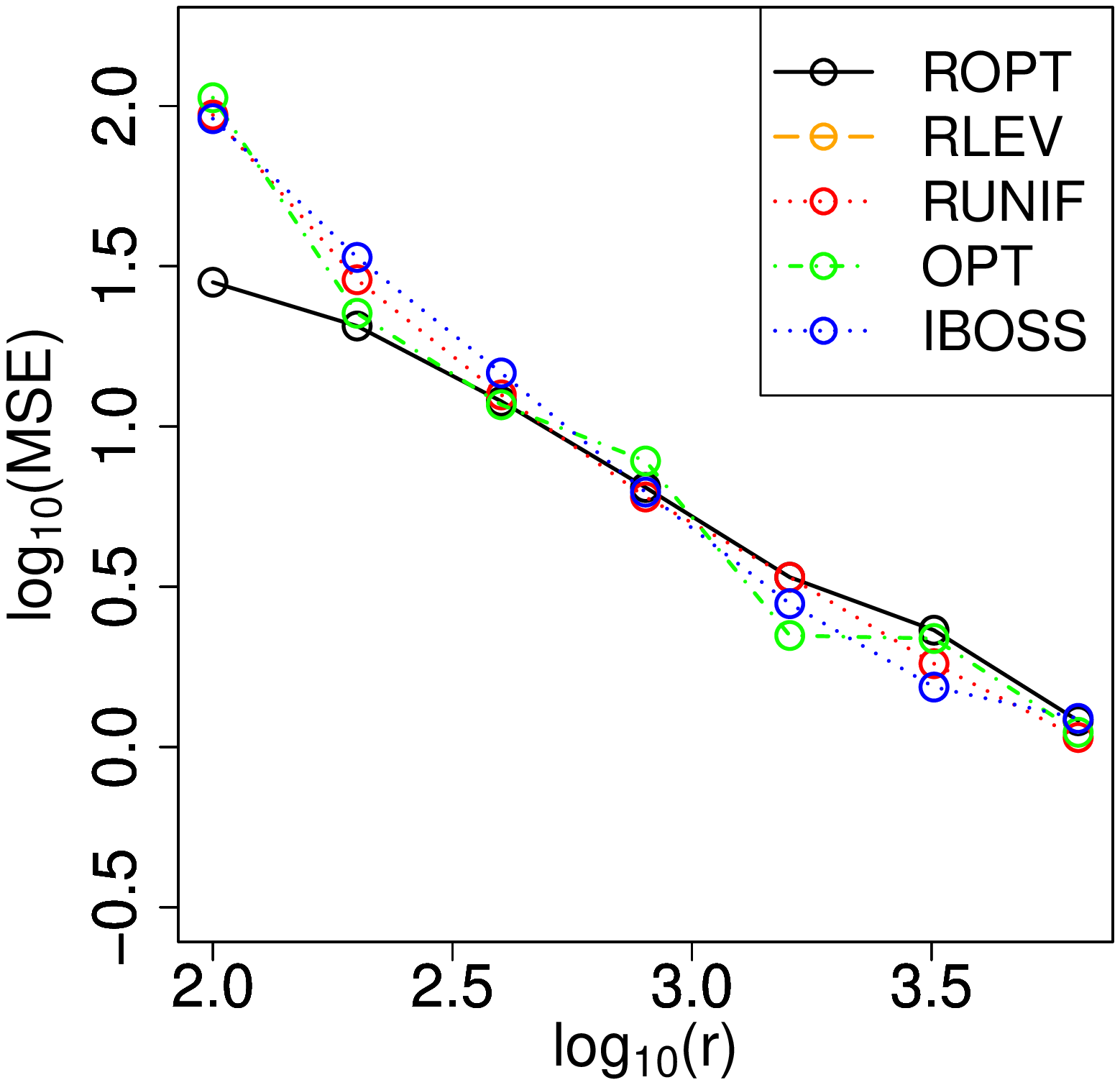}\vspace{4pt}
\includegraphics[width=1\linewidth]{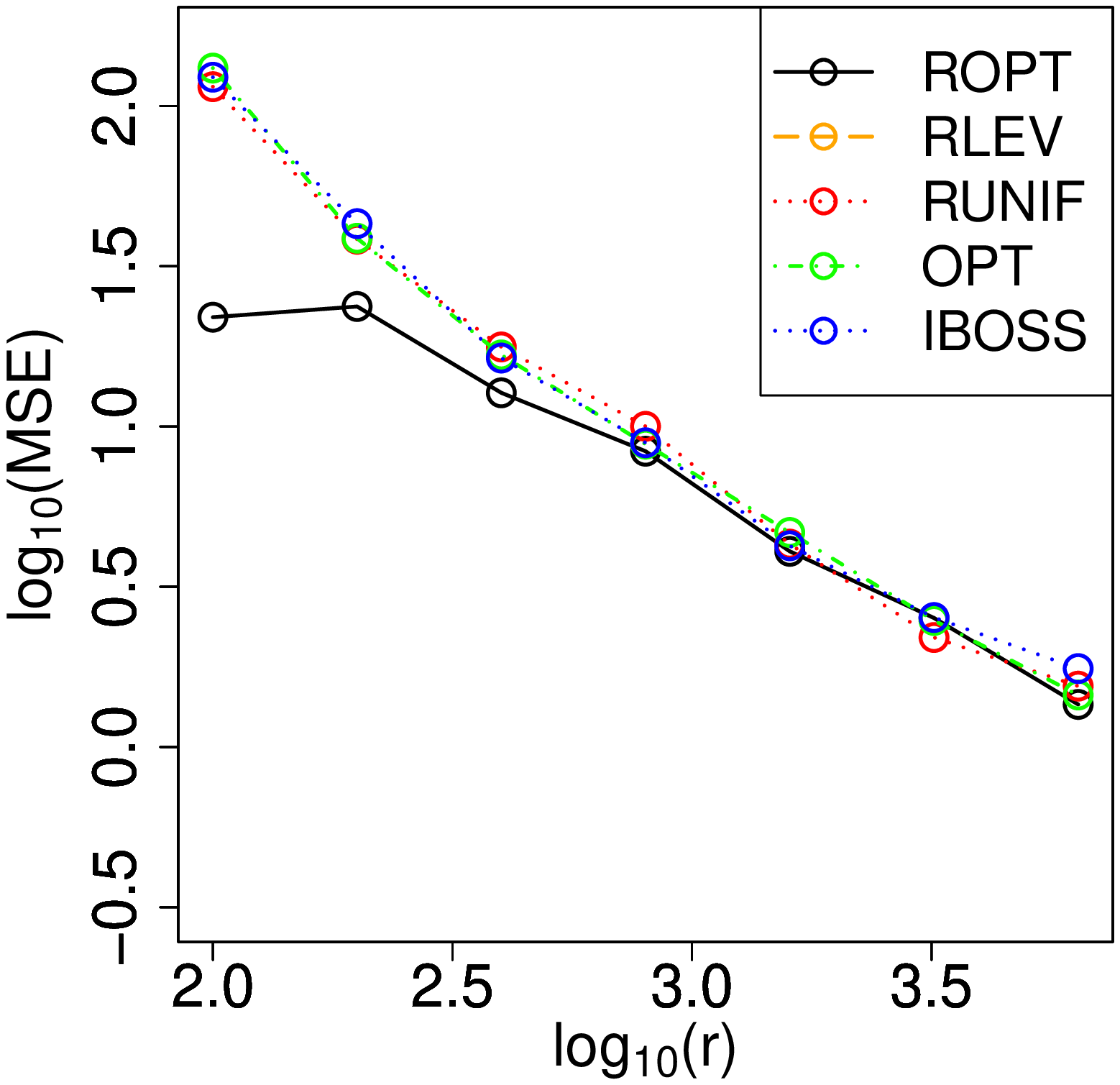}
\end{minipage}}
\subfigure{
\begin{minipage}[b]{0.31\linewidth}
\includegraphics[width=1\linewidth]{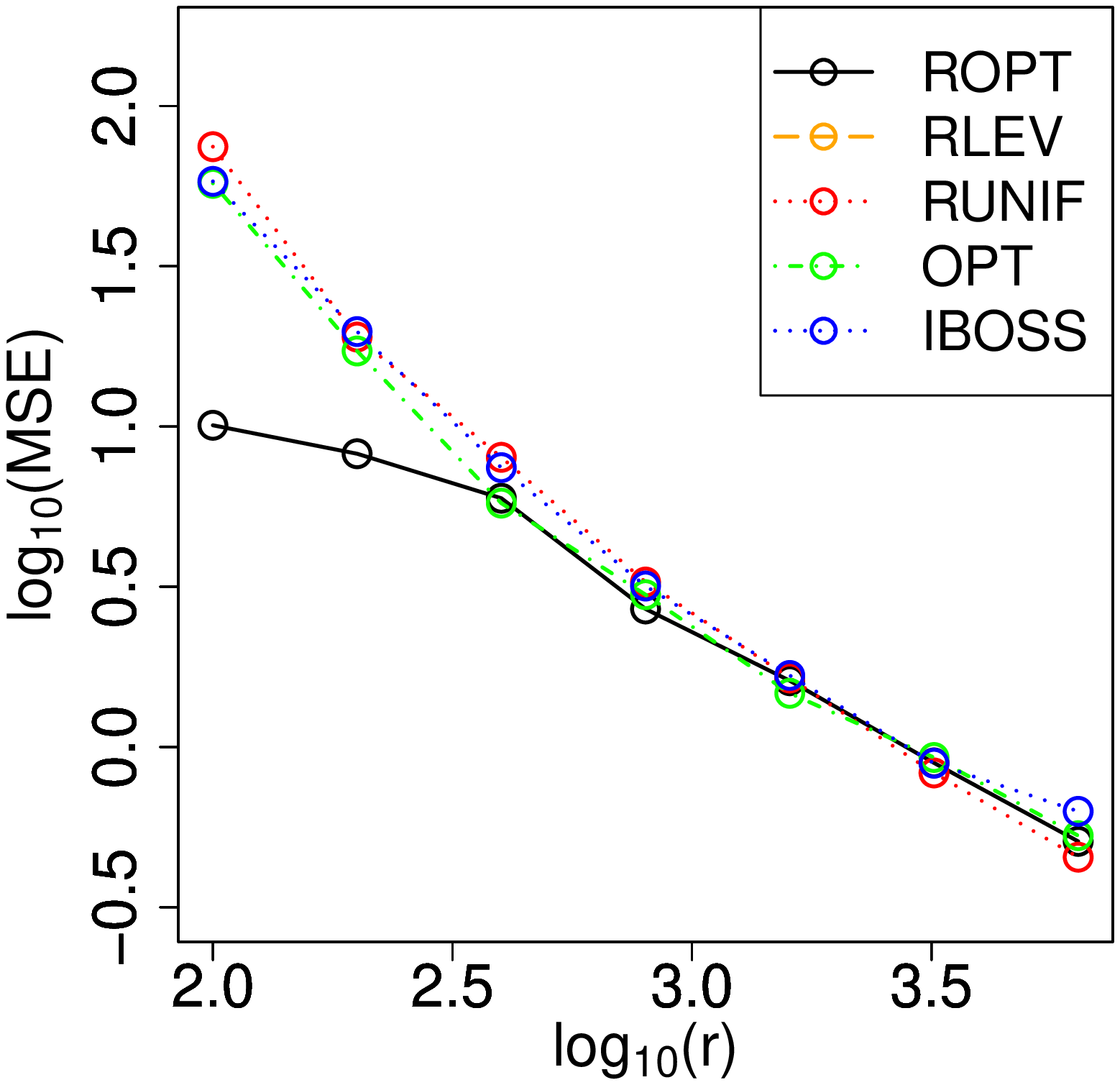}\vspace{4pt}
\includegraphics[width=1\linewidth]{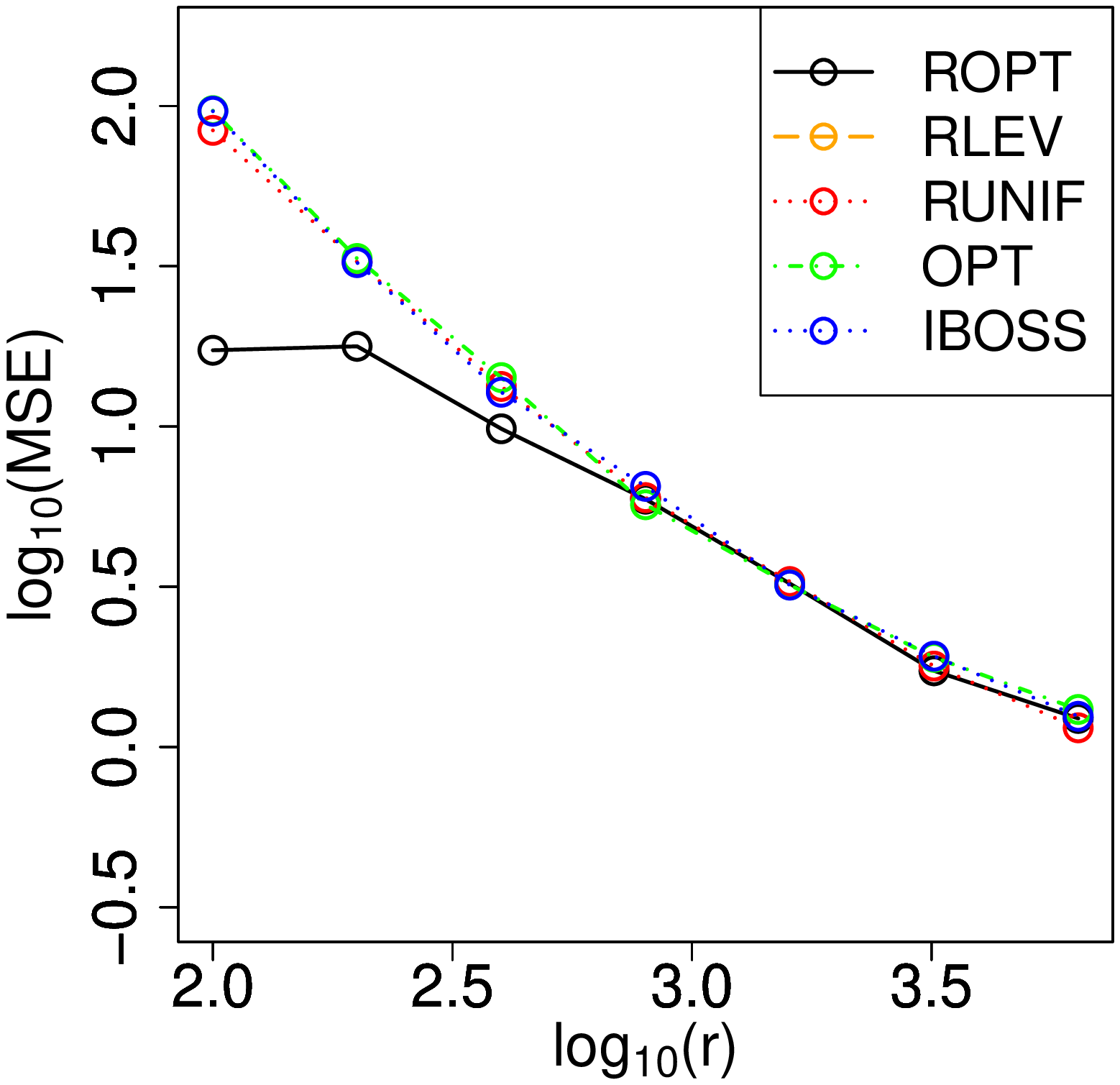}
\end{minipage}}
\caption{Comparison of different subsampling estimators: $x$-axis is the logarithm of the subsample size, $y$-axis is the logarithm of the mean squared error of the estimator $\widetilde{\boldsymbol{\beta}}$ compared with true $\boldsymbol{\beta}$.}
\label{simu_compnew}
\end{figure}
In Figure~\ref{simu_compnew}, our algorithm has the best performance among all the 6 cases when the subsample sizes are small or moderate, while all the methods have similar performance when the subsample sizes are large. First, our method has a great advantage when we use a small subsample. Second, by comparing the two rows of Figure~\ref{simu_compnew}, we can find our model preserves its superiority over other models for linear regression even in the cases where the true model favors less for introducing the ridge penalty. 

\section{Real data example}\label{real}
In times of information explosion, people are surrounded by a sea of news from various sources all day and night. For online media, it is critical for them to know what kind of news can attract the public attention, and hence the prediction of the popularity of the news becomes a trendy research topic. To raise the accuracy, numerical features from content, keywords, publish day and earlier popularity of news referenced in the article are extracted and then fed into a regression model to predict the share of the news. We use the open dataset of Online News Popularity Data Set on UC Irvine Machine Learning Repository\footnote{http://archive.ics.uci.edu/ml/datasets/Online+News+Popularity}, which was provided by \cite{fernandes2015proactive}.

The data was collected from Mashable, which is one of the largest news websites, from January 7, 2013, to January 7, 2015. It contains more than 39,000 articles in around 700 days. Except for the two non-predictable features, there is one response, the number of shares, and 58 predictive attributes concerning words, links, media, time, keywords, and natural language processing. Since the number of observations is huge and the number of features is also relatively large, it is prohibitive to allocate the memory for calculating the regression estimator. Therefore, we use the subsampling method to reduce the computation cost. The dataset is randomly divided into $70\%$ for training and $30\%$ for testing. Subsample sizes are set as $r=100,200,400,800,1600,3200,6400$. The design matrix $\mathbf{X}$ is standardized before being fed into the model. Each experiment is repeated $20$ times. Because the true regression coefficient $\boldsymbol\beta$ is unknown, we first compare our estimator $\widetilde{\boldsymbol{\beta}}$ with full-sample estimator $\widehat{\boldsymbol{\beta}}$ in terms of MSE.

\begin{figure}[h]
\centering
\includegraphics[width=0.45\linewidth]{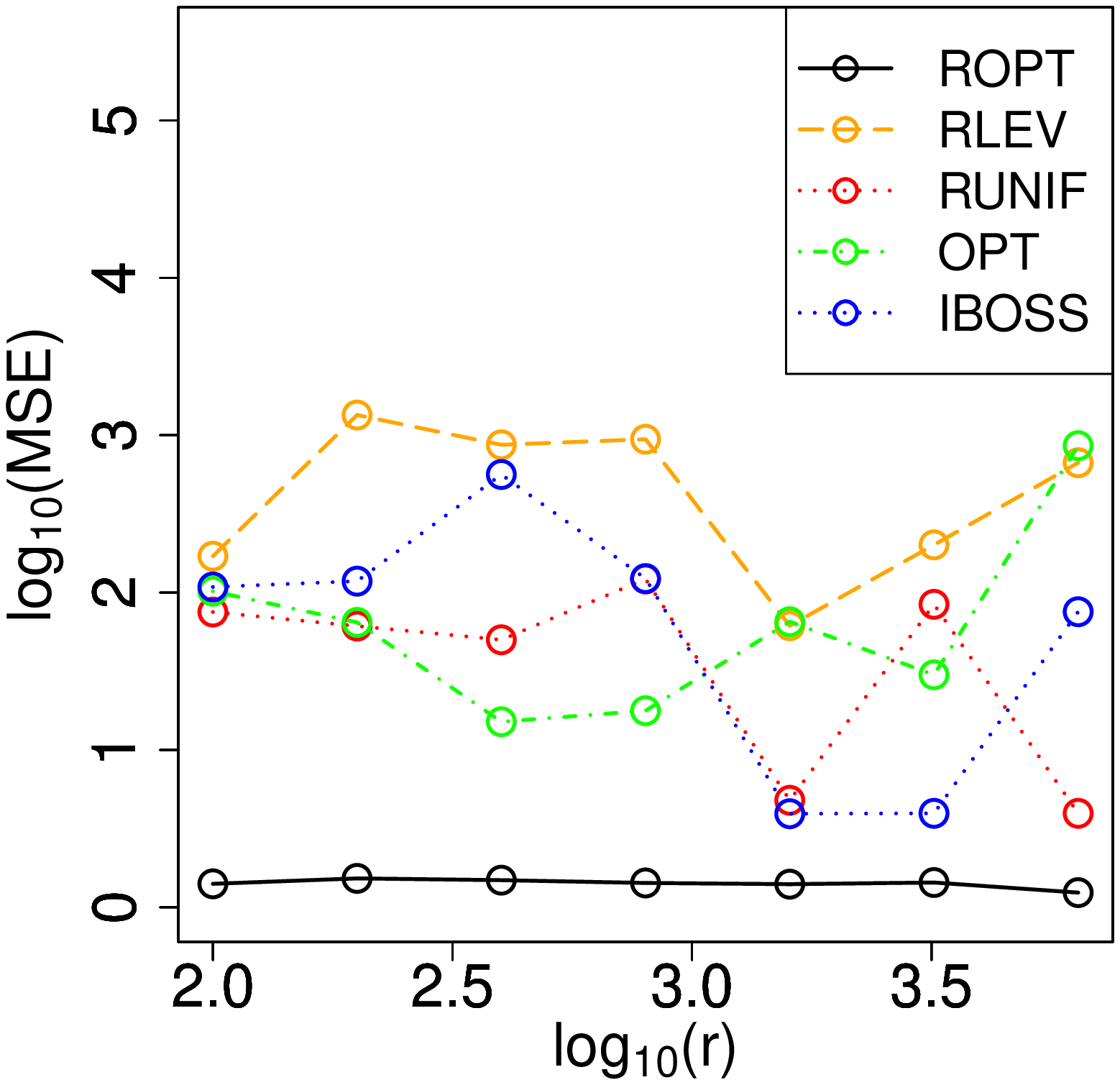}
\caption{Comparison of different subsampling estimators: $x$-axis is the logarithm of the subsample size, $y$-axis is the logarithm of the mean squared error of the estimator $\widetilde{\boldsymbol{\beta}}$ compared with full-sample estimator $\widehat{\boldsymbol{\beta}}$.}
\label{realdata}
\end{figure}

In Figure~\ref{realdata}, we plot the MSE of the estimators calculated by different methods. Our method has the best performance at various subsample sizes compared with the competing methods. 
Finally, to better compare the performance of different methods, we report the test error under various subsample sizes in Table~\ref{tab}. Our method keeps the advantage compared with other methods as the subsample size grows.

\begin{table}[h]
\centering
\begin{tabular}{c|c|c|c|c|c|c|c}
\toprule
r             & 100  & 200 & 400 & 800 & 1600 & 3200 & 6400   \\ 
\midrule
ROPT     & \textbf{0.007} & \textbf{0.043} & \textbf{0.015} & \textbf{0.012} & \textbf{0.013} & \textbf{0.015} & \textbf{0.029} \\ \hline
RLEV     & 1.560 & 1.754 & 0.594 & 1.032 & 0.555 & 0.516 & 0.285 \\ \hline
RUNIF & 1.317 & 1.389 & 1.062 & 1.160 & 0.350 & 0.653 & 0.289 \\ \hline
OPT & 1.819 & 1.571 & 0.978 & 0.633 & 1.684 & 1.412 & 2.761 \\ \hline
IBOSS & 1.425 & 1.030 & 1.892 & 1.261 & 0.402 & 0.349 & 1.137 \\ 
\bottomrule
\end{tabular}
\caption{The logarithm of the test error comparison under different subsample sizes.}
\label{tab}
\end{table}

\bibliographystyle{ims}
\bibliography{bib-power}

\end{document}